\documentclass[a4paper,UKenglish,cleveref, autoref, thm-restate]{lipics-v2021}

\usepackage{amssymb}
\usepackage{amsmath}
\usepackage{amsthm}
\usepackage{tikz}
\usepackage{pgfplots}
\usepackage{subcaption}
\usepackage[ruled]{algorithm2e}

\title{Compressibility Measures and Succinct Data Structures for Piecewise Linear Approximations} 

\author{Paolo Ferragina}{Department L'EMbeDS, Sant'Anna School of Advanced Studies, Pisa, Italy\\Department of Computer Science, University of Pisa, Italy}{paolo.ferragina@santannapisa.it}{https://orcid.org/0000-0003-1353-360X}{}

\author{Filippo Lari\footnote{Corresponding author}}{Department of Computer Science, University of Pisa, Italy}{filippo.lari@phd.unipi.it}{https://orcid.org/0009-0000-6817-6561}{}

\authorrunning{P. Ferragina and F. Lari}

\Copyright{Paolo Ferragina and Filippo Lari}



\category{} 

\relatedversion{} 

\funding{The work of P.F. was partially supported by the Alfred P. Sloan Foundation with the grant \#G-2025-25193 (\url{sloan.org}). The work of both authors has been supported by the NextGenerationEU -- National Recovery and Resilience Plan (Piano Nazionale di Ripresa e Resilienza, PNRR) -- Project: ``SoBigData.it - Strengthening the Italian RI for Social Mining and Big Data Analytics'' -- Prot.~IR0000013 -- Avviso n.~3264 del 28/12/2021; and by the European Union – Horizon 2020 Program under the scheme ‘‘INFRAIA-01-2018-2019 – Integrating Activities for Advanced Communities’’, Grant Agreement n. 871042, ‘‘SoBigData++: European Integrated Infrastructure for Social Mining and Big Data Analytics’’ (http://www.sobigdata.eu).}

\nolinenumbers

\EventEditors{John Q. Open and Joan R. Access}
\EventNoEds{2}
\EventLongTitle{42nd Conference on Very Important Topics (CVIT 2016)}
\EventShortTitle{}
\EventAcronym{}
\EventYear{2016}
\EventDate{December 24--27, 2016}
\EventLocation{Little Whinging, United Kingdom}
\EventLogo{}
\SeriesVolume{42}
\ArticleNo{23}

\begin{document}

\maketitle

\begin{abstract}
We study the problem of deriving compressibility measures for \emph{Piecewise Linear Approximations} (PLAs), i.e., error-bounded approximations of a set of two-dimensional {\em increasing} data points using a sequence of segments. Such approximations are widely used tools in implementing many \emph{learned data structures}, which mix learning models with traditional algorithmic design blocks to exploit regularities in the underlying data distribution, providing novel and effective space-time trade-offs. 

We introduce the first lower bounds to the cost of storing PLAs in two settings, namely {\em compression} and {\em indexing}. We then compare these compressibility measures to known data structures, and show that they are asymptotically optimal up to a constant factor from the space lower bounds. Finally, we design the first data structures for the aforementioned settings that achieve the space lower bounds plus small additive terms, which turn out to be {\em succinct} in most practical cases. Our data structures support the efficient retrieval and evaluation of a segment in the (compressed) PLA for a given $x$-value, which is a core operation in any learned data structure relying on PLAs.

As a result, our paper offers the first theoretical analysis of the maximum compressibility achievable by PLA-based learned data structures, and provides novel storage schemes for PLAs offering strong theoretical guarantees while also suggesting simple and efficient practical implementations.

\end{abstract}

\section{Introduction}\label{sec:introduction}

The rapid growth of data has led to the development of novel techniques for designing data structures that leverage regularities in the underlying data distribution to deliver faster or more succinct solutions to proper query operations. This novel line of research, known as \emph{learned data structures}~\cite{KraskaBCDP18}, integrates machine learning models with traditional algorithmic techniques to achieve improved space occupancy and query time performance.  

A widely used tool in the field of learned data structures is the so-called \emph{Piecewise Linear Approximation} (PLA), an error-bounded approximation of a two-dimensional set of points using a sequence of segments~\cite{ORourke81}. The common use of PLAs is due to their ability to capture linear trends in the input data, which frequently occur in many practical applications, while using little space and providing fast inference time. In this work, we focus on PLAs arising in two settings: (1) the \emph{compression} setting, where the input sequence of points is increasing in both coordinates but the $x$-values are also consecutive; and (2) the \emph{indexing} setting, where instead the $y$-values of the input sequence are consecutive. These two settings are general enough to cover a broad range of applications, including $rank$ and $select$ dictionaries~\cite{FerraginaFMV23,BoffaFV22}, time series compression~\cite{GuerraVBF25}, string indexing~\cite{WangTWC20}, one-dimensional indexing~\cite{KraskaBCDP18,FerraginaV20}, multi-dimensional indexing~\cite{NathanDAK20, DingNAK20}, range minimum queries~\cite{FerraginaL25}, monotone minimal perfect hashing~\cite{FerraginaLSV23}, and range filtering~\cite{VaidyaKCKMI22}, to mention a few.

Prior work concerning the study of theoretical properties of PLAs has primarily focused on providing {\em upper bounds} on the optimal number of segments $\ell$ as a function of the sequence length $n$ and the maximum allowed error $\varepsilon$, that is the maximum vertical distance between the $y$-coordinate of any input point and its \emph{covering} segment. Such a relationship is fundamental from a theoretical perspective since it can explain the surprising performance of learned data structures and also allow a theoretical comparison with their classical counterparts. In this context, it is worth mentioning the pioneering study of Ferragina et al.~\cite{FerraginaLV21} showing that in the indexing case under the assumption that the gaps between consecutive keys are independently and identically distributed random variables, with constant mean and variance, the optimal number of segments $\ell$ is $\Theta(n/\varepsilon^{2})$ with high probability. Later works~\cite{ZeighamiS23, CroquevielleYL025} provided better upper bounds, however, their PLAs were restricted to be formed by {\em horizontal segments}.



As far as {\em lower bounds} are concerned, we mention the work of Zeighami and Shahabi~\cite{ZeighamiS24} which was the first to limit from below the model size (e.g., the number of segments $\ell$ in a PLA or the number of neurons in a neural network) of any learned data structure achieving a given desired accuracy $\varepsilon$ for a wide range of problems arising in database systems such as one- and multi-dimensional indexing, cardinality estimation, and range-sum estimation.

In our work, we assume that the segments have \emph{arbitrary slopes} and the number of segments $\ell$ forming an optimal PLA is known (e.g., computed by \cite{ORourke81}), and thus provide the first information-theoretic lower bound to the storage complexity of that PLA in both the compression and the indexing settings (see Theorems~\ref{th:compr_setting} and~\ref{th:indexing_setting}). We then compare two representative data structures against these lower bounds, namely the LA-vector by Boffa et al.~\cite{BoffaFV22} (for the compression setting) and the PGM-index by Ferragina \& Vinciguerra~\cite{FerraginaV20} (for the indexing setting), showing that both are optimal up to a constant factor. Finally, we conclude our work by designing two novel data structures for the aforementioned compression and indexing settings, that provide the same asymptotic query time as the LA-vector and the PGM-index, but with small additive terms to the introduced lower bounds.


\section{Background}\label{sec:background}

We assume the transdichotomous word RAM model where the word size $w$ matches the logarithm of the problem size and arithmetic, logic, and bitwise operations, as well as access to $w$-bit memory cells, take $O(1)$ time. To simplify, we denote with $\log{x}$ the logarithm in base $2$ of $x$, and we assume that logarithms hide their ceiling and thus return integers.

Given a static sorted sequence $S = 1\leq x_{1} < x_{2} < \dots < x_{n} \leq u $ of $n$ elements drawn from a finite integer universe $U$ of size $u$, we consider the following fundamental operations:

\begin{itemize}
    \item $rank(S, x)$, which for $x \in U$ returns the number of $S$'s elements that are smaller or equal to $x$, i.e. $|\{x_{i} \in S\,|\,x_{i} \leq x\}|$.
    \item $select(S,i)$, which for $i \in \{1,2,\dots,n\}$, returns the element in $S$ of position $i$, i.e. $rank(S,x) = i$. We assume $select(S,0) = 0$.  
\end{itemize}

If $S$ is represented with its characteristic bit vector $B_{S}$, namely, a bit vector of length $u$ such that $B_{S}[i] = 1$ if $i \in S$, and $B_{S}[i] = 0$ otherwise, then supporting such operations is equivalent to support the $rank_{1}(B_{S},x)$, which yields the number of $1$s before position $x$, and $select_{1}(B_{S},i)$, which gives the position of the $i$-th $1$ in $B_{S}$. These operations have been studied extensively with many practical and theoretical results, and are at the core of any succinct or compressed data structure (see~\cite{Navarro16} for a comprehensive treatment). For our purposes, we briefly mention that they can be supported in $O(1)$ time while using only $o(u)$ additional bits on top of the underlying bit vector.

The information-theoretic lower bound to store $S$ is $\mathcal{B}(u,n) = \log{\binom{u}{n}}$. Such value is related to the zero-order empirical entropy of its characteristic bit vector $B_{S}$, defined as $u\mathcal{H}_{0}(B_{S}) = n\log{\frac{u}{n}} + (u-n)\log{\frac{u}{u-n}}$. In fact, with simple arithmetic manipulations one obtains $\mathcal{B}(u,n) = u\mathcal{H}_{0}(B_{S}) - O(\log{u})$. The best upper bound to represent $S$ with the minimum redundancy with respect to $\mathcal{B}(u,n)$, and while still being able to perform $rank$ and $select$ in optimal constant time, is given by the following data structure due to P\u{a}tra\c{s}cu~\cite{Patrascu08}:

\begin{theorem}[Succincter]\label{th:succincter}
    For any $c>0$, a monotone sequence $S$ of $n$ positive integers drawn from a universe of size $u$ can be stored using $\mathcal{B}(u,n) + O(u/\log^{c}(u/c)) + O(u^{3/4}\: poly\log{u})$ bits and support $rank$ and $select$ in $O(c)$ time.
\end{theorem}

Theorem~\ref{th:succincter} implies constant time $rank$ and $select$ operations in $\mathcal{B}(u,n) + O(u/\log^{c}{u})$ space, which essentially matches the lower bound of P\u{a}tra\c{s}cu \& Viola~\cite{PatrascuV10}. To further reduce the space redundancy on top of $\mathcal{B}(u,n)$, one must allow for slower $rank$ or $select$ operations. A prominent example is given by the Elias-Fano encoding scheme~\cite{Fano71, Elias74}, which attains the following space-time trade-off for the aforementioned operations~\cite{Navarro16}:

\begin{theorem}[Elias-Fano encoding]\label{th:elias_fano}
    A monotone sequence $S$ of $n$ positive integers drawn from a universe of size $u$ can be stored using $2n + n\log{\frac{u}{n}} + o(n)$ bits while supporting $select$ in $O(1)$ time and rank in $O(\min\{\log{\frac{u}{n}}, \log{n}\})$ time.
\end{theorem}

Indeed, the space usage of Theorem~\ref{th:elias_fano} can be expressed in terms of $\mathcal{B}(u,n)$ as follows:
\begin{equation}\label{eq:ef_entropy}
    2n + n\log{\frac{u}{n}}+o(n) = u\mathcal{H}_{0}(B_{S}) + O(n) = \mathcal{B}(u,n) + O(n+\log{u})
\end{equation}

Thus obtaining a lower redundancy term with respect to Theorem~\ref{th:succincter} whenever $S$ is very sparse, which, however, requires giving up the constant time $rank$ operation. The literature offers other trade-offs (e.g.,~\cite{OkanoharaS07}) that could be adopted in our solution, leading to corresponding time-space bounds. However, for the sake of presentation, in the remainder of the paper we focus on the results provided in Theorems~\ref{th:succincter} and~\ref{th:elias_fano}. 


A new line of research has recently started to investigate novel methods for designing algorithms and data structures that leverage regularities in the underlying data to deliver faster or more succinct solutions to a wide variety of problems. Most of these methods rely on a technique known as \emph{Piecewise Linear Approximation} (PLA), which involves approximating a sequence of two-dimensional points using a set of segments.

\begin{definition}\label{def:pla}
    Given a set of $n$ two-dimensional data points $P \subseteq \mathbb{R}^{2}$, a Piecewise Linear Approximation of $P$ with maximum error $\varepsilon \geq 1$ is the smallest set of non-overlapping segments $\{s_{1},s_{2},\dots,s_{\ell}\}$ that entirely covers $P$, ensuring that every point in $P$ lies within a maximum vertical distance $\varepsilon$ from its corresponding segment.
\end{definition}

A PLA can always be computed optimally, hence producing the minimum number of segments $\ell$ in time proportional to the sequence length using an algorithm by O'Rourke~\cite{ORourke81}.

Here, we focus on PLAs computed from monotone sequences of two-dimensional data points, in which a point $(x_i,y_i)$ is smaller than $(x_{i+1},y_{i+1})$ because it has a smaller abscissa $x_i < x_{i+1}$ and a smaller ordinate $y_i < y_{i+1}$. In particular, we consider two specific scenarios: (1) the compression setting, where PLAs are computed on sequences having consecutive $x$-values (i.e., $x_{i+1} = x_i +1$, typically $x_1=1$), and (2) the indexing setting, in which PLAs are computed on sequences having consecutive $y$-values (i.e. $y_{i+1} = y_i +1$, typically $y_1=1$). As already mentioned in the introduction, despite the names, such scenarios are quite general and encompass a wide range of applications. To get a practical sense of the usage of PLAs in the design of learned data structures, we briefly outline two prominent examples.

In the design of learned and compressed $rank$ and $select$ dictionaries~\cite{BoffaFV22}, the input sequence $S = 1 \leq x_{1} < x_{2} < \dots < x_{n} \leq u$, is first mapped to a Cartesian plane representation where the $x$-axis corresponds to the sequence of positions $1, 2, \dots, n$, while the $y$-axis represents the sequence values, namely $x_i$s. This transformation converts $S$ into a set of monotonically increasing two-dimensional points: $\{(i,x_{i})\,|\,x_{i} \in S\}$. The key innovation of this approach is that the Cartesian plane representation reveals inherent regularities in the underlying sequence $S$. A PLA is then computed from such a sequence of points, and compression is achieved by storing for each point $(i,x_{i})$ the absolute difference between $x_{i}$ and the approximation of the segment for position $i$ using just $\log{(2\varepsilon+1)}$ bits. Performing a $select(S,i)$ operation then reduces to first finding the segment covering the abscissa value $i$, computing the approximate sequence value from such segment in $i$, and finally adding the correction term stored in position $i$ to retrieve the exact value. A $rank(S,x)$ operation can be naively solved by binary searching on the $select$ values, even if faster alternatives do exist~\cite{BoffaFV22}. 

In the design of learned one-dimensional indexes~\cite{FerraginaV20}, given the input sequence $S = 1 \leq x_{1} < x_{2} < \dots < x_{n} \leq u$, one performs the same mapping above onto a Cartesian plane representation. In this case, the $x$-axis contains the sequence values, while the $y$-axis corresponds to their respective positions $1, 2, \dots, n$. Again, this transformation converts $S$ into a set of monotonically increasing two-dimensional data points: $\{(x_{i},i)\,|\; x_{i}\in S\}$. The data structure then consists of a PLA computed from such a set of points and a search for a given key $x$ proceeds in three steps. First, the segment covering $x$ is identified through a binary search on the first keys covered by each segment. Second, the approximate position $i$ of $x$ is obtained from such a segment. Third, since the underlying sequence is monotonically increasing and the approximate position $i$ of $x$ is correct up to an error $\varepsilon$, it is \emph{corrected} via a (binary) search in the interval $[i-\varepsilon,i+\varepsilon]$. Such queries can be sped up by applying the same construction recursively to the sequence of first keys covered by each segment until a single one remains (see e.g. \cite{FerraginaV20,KraskaBCDP18}). This approach implements the index as a hierarchical structure, forming a tree of linear models. Consequently, searches involve a top-down traversal of the tree, trading faster queries for an increased space usage.

\section{Related Work}\label{sec:related_work}

In the following, we review existing methods for the lossless compression of PLAs, focusing on the $predict(x)$ operation, which, given a value lying on the $x$-axis, identifies the segment covering $x$ and computes the corresponding $y$-value \emph{predicted} by that segment.

Ferragina \& Vinciguerra introduced the PGM-index~\cite{FerraginaV20}, the first solution in the indexing case based on the optimal PLA computed by the O'Rourke algorithm~\cite{ORourke81}, i.e., the one consisting of the minimum number of segments $\ell$ providing an $\varepsilon$-approximation of the input points. In particular, they designed two compression methods for the segments. The first one hinges on the observation that, in the indexing setting, the intercepts are monotonically increasing. Denoting with $\beta_{i}$ the intercept of the $i$-th segment and with $y_{i}$ its first covered $y$-value, this easily follows since $\beta_{i} \in [y_{i}-\varepsilon,y_{i}+\varepsilon]$ and $\beta_{i+1} \in [y_{i+1}-\varepsilon,y_{i+1}+\varepsilon]$, but $y_{i}+\varepsilon \leq y_{i+1}-\varepsilon$ as in this particular case a segment covers at least $2\varepsilon$ keys~\cite[Lemma 2]{FerraginaV20}. Exploiting this, they convert the monotone sequence of floating-point intercepts into integers and store them with the data structure of Okanohara \& Sadakane~\cite{OkanoharaS07}. But, they stored uncompressed the first covered keys and slopes of each segment. For the sake of comparison, we provide the space usage for the PLA composing just the first layer of the PGM index. Since the universe of the underlying sequence is $u$, and each slope can be expressed as a rational number with a numerator and denominator respectively of $\log{n}$ and $\log{u}$ bits, they obtained the following result which we specialize via two approaches, one based on binary search (as proposed in~\cite{FerraginaV20}), and one using the data structure of P\u{a}tra\c{s}cu (see Theorem~\ref{th:succincter}):

\begin{lemma}[PGM-index~\cite{FerraginaV20}]\label{lem:pgm_space_time}
    The PLA composing the first layer of the PGM-index can be stored using:
    \begin{itemize}
        \item $\ell \, (1.92 + \log{\frac{n^{2}}{\ell}} + 2\, \log{u} + o(1))$ bits, while supporting the $predict$ operation in $O(\log{\ell})$ time via binary search.
        \item $\ell \, (1.92 + \log{\frac{n^{2}}{\ell}} + \, \log{u} + o(1)) + \log{\binom{u}{\ell}} + O(u/\log^c u)$ bits for any constant $c > 0$, while supporting the $predict$ operation in $O(c)$ time.
    \end{itemize}
\end{lemma}

The second compression method targets the segment's slopes, and is based on the observation that the O'Rourke algorithm does not compute a single segment but a whole family of segments whose slopes identify an interval of reals. The compressor then works greedily (and optimally) by trying to assign the same slope to most segments of the PLA, thus reducing the number of slopes to be fully represented. This is very effective on some real-world datasets, however, the authors did not provide an analysis of its space occupancy.

Moving to the compression setting, Boffa et al. introduced the LA-vector~\cite{BoffaFV22}, the first learned and compressed $rank$ and $select$ dictionary based on the optimal PLA. Their compression scheme for segments hinges on the assumption that the sequence of intercepts is monotonically increasing, which is not necessarily the case in general. This is because, in the compression setting, segments may cover fewer than $2\varepsilon$ points, as the $y$-value can increase by more than one between consecutive $x$-values. Ferragina \& Lari~\cite{FerraginaL25} recently showed how to waive that assumption with a negligible impact on space occupancy, obtaining the following bound (we drop the cost of $cn$ bits for storing the correction vector $C$ in Theorem 3.2 of~\cite{BoffaFV22}, and we add the solution based on the data structure of P\u{a}tra\c{s}cu):

\begin{lemma}[LA vector, Theorem 2.3 in~\cite{BoffaFV22} and Theorem 5 in~\cite{FerraginaL25}]\label{lem:la_space_time}
    The PLA of the LA-vector can be stored using  
    
    \begin{itemize}
        \item $\ell\, (2\log{\frac{u}{\ell}} +  \log{\frac{n}{\ell}} + 6 + 2\log{(2\varepsilon+1)} + o(1))$ bits, while supporting the $predict$ operation in $O(\log{\ell})$ time via binary search.
        \item $\ell\, (2\log{\frac{u}{\ell}} + 4 + 2\log{(2\varepsilon+1)} + o(1)) + \log{\binom{n}{\ell}} + O(n/\log^c n)$ bits for any constant $c > 0$, while supporting the $predict$ operation in $O(c)$.
    \end{itemize}
\end{lemma}

Lastly, Abrav and Medvedev proposed the PLA-index~\cite{AbrarM24} in a Bioinformatics application where the $x$- and $y$-values of the underlying sequence are increasing but not necessarily contiguous. They took an orthogonal approach to PLA compression by modifying the O'Rourke algorithm, enforcing each segment to start from the ending point of the previous one. This allows them to derive a more compact encoding, which may produce a suboptimal number of segments $\ell' \geq \ell$. As their approach does not guarantee optimality and targets different sequence types, we defer comparison to the journal version.

We conclude by noticing that none of these results focused on studying lower bounds to the cost of storing PLAs (possibly in compressed form). Therefore, the first novel contribution of our work lies in deriving the first compressibility measures for PLAs in the two aforementioned settings. It will turn out that the storage schemes for PLAs of the LA-vector and the PGM-index are asymptotically optimal up to a constant factor. Hence, our second contribution is to design the first data structures for the compression and indexing settings that achieve the space lower bounds plus small additive terms, which are {\em succinct} in most practical cases. 

\section{Compressibility Measures for PLAs}\label{sec:compr_measures_plas}

\begin{figure}[t]
    \centering
    \begin{subfigure}{0.47\textwidth}
        \centering
        \begin{tikzpicture}[scale=0.8]
            \begin{axis}[
                xlabel={\textit{positions}},
                ylabel={\textit{values}},
                ymin=0, ymax=11,
                xmin=0, xmax=11,
                xtick=\empty,
                ytick=\empty,
                xtick={1,3.75,5,9},
                xticklabels={$x_{1} = 1$, \small{$x_{2}-1$}, $x_{2}$, $x_{3}$},
                ytick={5, 6.2, 7, 8},
                yticklabels={$y_{2}$, $\beta_{2}$, $y'_{2}$, $\gamma_{2}$},
                major tick length=0pt,
            ]

            \addplot[domain=1:4, samples=2, thick, line width=1.5pt] {0.65*(x-1)+1}; 
            \addplot[domain=5:8, samples=2, thick, line width=1.5pt] {0.6*(x-5)+6.2}; 
            \addplot[domain=9:11, samples=2, thick, line width=1.5pt] {2*(x-9)+6.1}; 
        
            \draw[dotted, thick] (axis cs:1,1) -- (axis cs:1,0);

            \draw[dotted, thick] (axis cs:5,4) -- (axis cs:5,0);
            \draw[dotted, thick] (axis cs:0,5) -- (axis cs:5,5);
            \draw[dotted, thick] (axis cs:0,6.2) -- (axis cs:5,6.2); 

            \draw[dotted, thick] (axis cs:0,8) -- (axis cs:8,8); 
            \draw[dotted, thick] (axis cs:0,7) -- (axis cs:8,7);
            
            \draw[dotted, thick] (axis cs:9,0) -- (axis cs:9,8);
            \draw[dotted, thick] (axis cs:4,0) -- (axis cs:4,2.95);
        
            \draw[fill=black] (axis cs:5,5) circle (1.5pt);
            \draw[fill=black] (axis cs:8,7) circle (1.5pt); 
            \draw[fill=black] (axis cs:9,7.5) circle (1.5pt);
        
            \draw[thick] (axis cs:5, 3.5) -- (axis cs:5, 6.5); 
            \draw[thick] (axis cs:4.9, 3.5) -- (axis cs:5.1, 3.5);
            \draw[thick] (axis cs:4.9, 6.5) -- (axis cs:5.1, 6.5);
        
            \node[anchor=east] at (axis cs:5, 6.5) {\tiny{$y_{2} + \varepsilon$}};
            \node[anchor=east] at (axis cs:5, 3.5) {\tiny{$y_{2} - \varepsilon$}};

            \draw[thick] (axis cs:8, 5.5) -- (axis cs:8, 8.5); 
            \draw[thick] (axis cs:7.9, 5.5) -- (axis cs:8.1, 5.5);
            \draw[thick] (axis cs:7.9, 8.5) -- (axis cs:8.1, 8.5);
        
            \node[anchor=east] at (axis cs:8, 8.5) {\tiny{$y'_{2} + \varepsilon$}};
            \node[anchor=east] at (axis cs:8, 5.5) {\tiny{$y'_{2} - \varepsilon$}};

            \draw[thick] (axis cs:9, 6) -- (axis cs:9, 9); 
            \draw[thick] (axis cs:8.9, 6) -- (axis cs:9.1, 6);
            \draw[thick] (axis cs:8.9, 9) -- (axis cs:9.1, 9);
            
            \end{axis}
        \end{tikzpicture}
        \subcaption{}
        \label{fig:pla_compression}
    \end{subfigure}
    \hfill
    \begin{subfigure}{0.47\textwidth}
        \centering
        \begin{tikzpicture}[scale=0.8]
            \begin{axis}[
                xlabel={\textit{values}},
                ylabel={\textit{positions}},
                ymin=0, ymax=15,
                xmin=0, xmax=15,
                xtick=\empty,
                ytick=\empty,
                xtick={1,5,11,13},
                xticklabels={$x_{1}$, $x_{2}$, $x'_{2}$, $x_{3}$},
                ytick={3.75, 5, 12, 13},
                yticklabels={$\beta_{2}$, $y_{2}$, $y_{3}-1$, $\gamma_{2}$},
                major tick length=0pt,
            ]
        
            \addplot[domain=1:3, samples=2, thick, line width=1.5pt] {0.65*(x-1)+1};
            \addplot[domain=5:11, samples=2, thick, line width=1.5pt] {1.5*(x-5)+3.75};
            \addplot[domain=13:15, samples=2, thick, line width=1.5pt] {0.3*(x-13)+14};
        
            \draw[dotted, thick] (axis cs:1,1) -- (axis cs:1,0);
            \draw[dotted, thick] (axis cs:3,0) -- (axis cs:3,2.3);
            
            \draw[dotted, thick] (axis cs:5,4) -- (axis cs:5,0);
            \draw[dotted, thick] (axis cs:0,3.75) -- (axis cs:5,3.75);
            \draw[dotted, thick] (axis cs:0,12.75) -- (axis cs:11,12.75); 
            \draw[dotted, thick] (axis cs:13,0) -- (axis cs:13,14); 

            \draw[dotted, thick] (axis cs:0,12) -- (axis cs:11,12); 
            \draw[dotted, thick] (axis cs:0,5) -- (axis cs:5,5);
        
            \draw[fill=black] (axis cs:5,5) circle (1.5pt); 
            \draw[fill=black] (axis cs:11,12) circle (1.5pt); 
        
            \draw[thick] (axis cs:5, 3.5) -- (axis cs:5, 6.5); 
            \draw[thick] (axis cs:4.9, 3.5) -- (axis cs:5.1, 3.5);
            \draw[thick] (axis cs:4.9, 6.5) -- (axis cs:5.1, 6.5);
        
            \node[anchor=east] at (axis cs:5, 6.75) {\tiny{$y_{2} + \varepsilon$}};
            \node[anchor=east] at (axis cs:5, 3.25) {\tiny{$y_{2} - \varepsilon$}};
        
            \draw[thick] (axis cs:11, 13.5) -- (axis cs:11, 10.5); 
            \draw[thick] (axis cs:10.9, 13.5) -- (axis cs:11.1, 13.5);
            \draw[thick] (axis cs:10.9, 10.5) -- (axis cs:11.1, 10.5);
        
            \node[anchor=north] at (axis cs:11, 14.5) {\tiny{$y_{3} - 1 + \varepsilon$}};
            \node[anchor=south] at (axis cs:11, 9.25) {\tiny{$y_{3} - 1 - \varepsilon$}};

            \draw[dotted, thick] (axis cs:11, 0) -- (axis cs:11, 10.5);
            
            \end{axis}
        \end{tikzpicture}
        \subcaption{}
        \label{fig:pla_indexing}
    \end{subfigure}
    \caption{The key parameters of PLAs arising in the two considered settings. Figure~\ref{fig:pla_compression} shows an example in the compression case. The values on the $x$-axis are consecutive, the first segment always begins at position $1$, and each segment starts at an $x$-value immediately after the ending of the previous one. Intercepts and the last $y$-values of each segment stay within a range of size $2\varepsilon+1$ from the first and last covered values of the sequence. Notice how the intercepts do not necessarily form an increasing sequence. Figure~\ref{fig:pla_indexing} considers the indexing case, which is similar to the compression setting except that the values on the $y$-axis are consecutive. Hence, the last covered $y$-value of a segment is implicitly determined by the first covered $y$-value of the next one, and the intercepts $\beta_{i}$s form a monotone sequence as $y_{i} + \varepsilon \leq y_{i+1}-\varepsilon$ since a segment covers at least $2\varepsilon$ keys~\cite[Lemma 2]{FerraginaV20}.}
    \label{fig:pla_compression_and_indexing}
\end{figure}
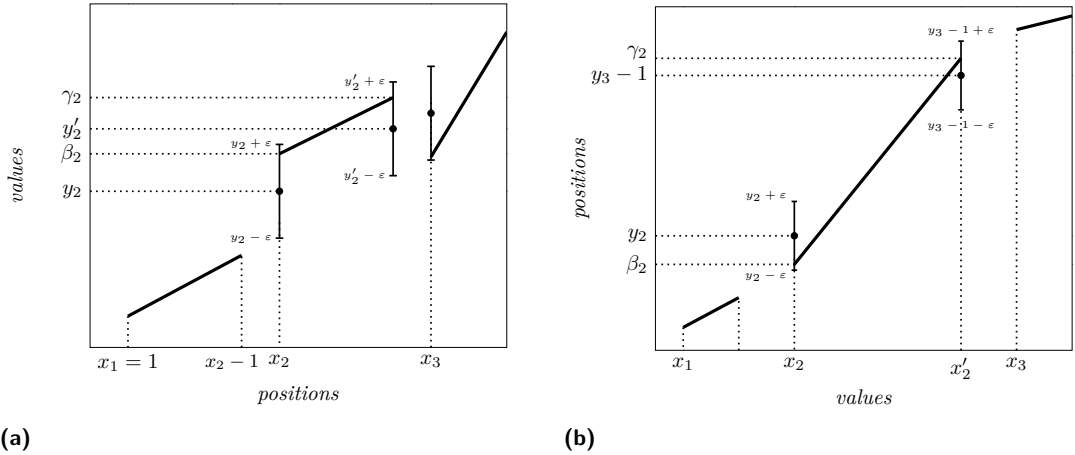

We begin by observing that, in general, the segments of a PLA, whether in the compression or indexing setting, have floating-point components. As discussed in Section~\ref{sec:related_work}, some existing compression schemes assume that segments either start or end at an integer $y$-coordinate to exploit integer compression techniques. In our context, we work under the same assumption for two reasons: first, when establishing a lower bound on the space required to represent a PLA, it is important to note that the class of PLAs with floating-point components is larger than that of PLAs with integer components. Consequently, any lower bound proven for the latter also applies to the former. Second, as proved by Ferragina \& Lari~\cite{FerraginaL25} this reduction has a negligible impact on the approximation error of the resulting PLA, and in fact, it only increases that error by a constant $3$. Therefore, without loss of generality, we can safely restrict our analysis to PLAs having segments with integer starting and ending ordinates.

Under this assumption, given the optimal number of segments $\ell$, the approximation error $\varepsilon$, the universe size $u$, and the length of the sequence $n$, we consider the problem of counting the number of possible PLAs that can be constructed using these parameters, in both the compression and indexing cases. Then, by a simple information-theoretic argument,  taking the logarithm of such quantities gives a lower bound on the number of bits required to represent any PLA in the considered settings, and this provides a compressibility measure that any compression scheme should strive to achieve. Starting from PLAs arising in the compression setting, we prove the following theorem:  

\begin{theorem}\label{th:compr_setting}
    Let $\mathcal{C}(\ell,\varepsilon,u,n, \mathbf{y})$ be the set of PLAs of $\ell$ segments having maximum error $\varepsilon$, covering a monotone sequence of $n$ elements drawn from a universe of size $u$ in the compression setting, and having $\mathbf{y} = y_1 \leq y_2 \leq \dots \leq y_{\ell}$ as the first covered $y$-values. Then:
    \begin{equation*}
        \lvert \mathcal{C}(\ell,\varepsilon,u,n,\mathbf{y}) \rvert = \binom{n-\ell-1}{\ell-1}\binom{u+\ell-1}{\ell}\left(\prod_{i=1}^{\ell-1}(y_{i+1}-y_{i}+1)\right)(2\varepsilon+1)^{2\ell}
    \end{equation*}
\end{theorem}


\begin{proof}
    Let us denote with $x_i$ the first position covered by the $i$-th segment, and let us start by constraining each $x_{i}$ in such a way as to obtain a valid PLA with error $\varepsilon$. Since the first segment always starts from the first position, it is $x_1 = 1$. Each segment covers at least two elements because there is always a segment connecting two points without incurring any error (i.e., the one connecting them), hence $x_{i+1}-x_{i} \geq 2$ (recall that in the compression setting the $x$-values are strictly increasing and contiguous). 
        
    To ease our counting, we introduce the following change of variables $\bar{x}_{i} = x_{i} - (i-1)$. In this way, $\bar{x}_{2} = x_2 - 1 \geq 2 $ and $\bar{x}_{\ell} = x_{\ell} - (\ell-1) \leq n - \ell$, since $x_{2} \geq 3$ and $x_{\ell} \leq n-1$. Additionally, $\bar{x}_{i+1} - \bar{x}_{i} \geq 1$ (since the $\bar{x}_{i}$s are distinct) and thus $x_{i+1} - x_{i} = \bar{x}_{i+1} - \bar{x}_{i} + 1 \geq 2$. Therefore, choosing the first position of each segment while simultaneously satisfying the previous constraint on the $x_{i}$s, is equivalent to choosing $\ell-1$ distinct values from $\{2, 3, \dots, n-\ell\}$. The total number of ways to perform this choice is then $\binom{n-\ell-1}{\ell-1}$.


    Let $\beta_{i}$ and $y_{i}$ denote, respectively, the intercept of the $i$-th segment and its first covered $y$-value (see Figure~\ref{fig:pla_compression}). Choosing $\beta_{i}$ is equivalent to first selecting $y_{i}$ and then choosing a value within the interval $[y_{i}-\varepsilon,\,y_{i}+\varepsilon]$, which contains $2\varepsilon+1$ elements. In the compression setting, the input sequence may contain repeated values, so consecutive segments can share the same first covered $y$-value (i.e., $y_{i+1}=y_{i}$). Thus, choosing the first $y$-value of each of the $\ell$ segments amounts to selecting an increasing sequence of $\ell$ elements from $\{1,2,\dots,u\}$. The number of such choices is $\binom{u+\ell-1}{\ell}$. Since each $\beta_{i}$ can then be any value of the interval $[y_{i}-\varepsilon,\,y_{i}+\varepsilon]$, the total number of possible selections for the $\ell$ intercepts is $\binom{u+\ell-1}{\ell}(2\varepsilon+1)^{\ell}$.
    

    Lastly, let $\gamma_{i}$ be the $y$-value given by the $i$-th segment evaluated on its last covered position (see Figure~\ref{fig:pla_compression}). We perform the same reasoning as for the intercepts. Let $y'_{i}$ be the value of the sequence on the last position covered by the $i$-th segment. Choosing $\gamma_{i}$ is equivalent to first selecting $y'_{i}$ and then choosing any of the values inside $[y'_{i}-\varepsilon,y'_{i}+\varepsilon]$. Now, since we are dealing with monotonically increasing sequences, to obtain a valid PLA, we need that $y'_{i} \geq y_{i}$ and $y'_{i} \leq y_{i+1}$ for every $1\leq i < \ell$, notice that $y'_{\ell} = u$. Consequently, each $y'_i$ can be independently chosen from the interval $[y_i, y_{i+1}]$, and each $\gamma_{i}$ is then selected inside the interval $[y'_{i}-\varepsilon,y'_{i}+\varepsilon]$, therefore, the total number of ways to choose the $\ell$ last $y$-values of each segment is $\left(\prod_{i=1}^{\ell-1}(y_{i+1}-y_{i}+1)\right)(2\varepsilon+1)^{\ell}$.
    
\noindent Combining these quantities, our claim follows.
\end{proof}

Theorem~\ref{th:compr_setting} suggests the following lower bound on the number of bits that any lossless compressor for PLAs in the compression setting must output.

\begin{corollary}\label{cor:compr_setting}
    The minimum number of bits to represent any PLA from $\mathcal{C}(\ell, \varepsilon, u, n, \mathbf{y})$ is:
    \begin{align*}
        \mathcal{B}_{\mathcal{C}}(\ell, \varepsilon, u, n, \mathbf{y}) &= \log{\binom{n-\ell-1}{\ell-1}} + \log{\binom{u+\ell-1}{\ell}}\\ &+ \sum_{i=1}^{\ell-1}\log{(y_{i+1}-y_{i}+1)} + 2\ell\log{(2\varepsilon+1)}
    \end{align*}
\end{corollary}

Next, we focus on PLAs arising in the indexing setting. Due to space limitations, we present only the main result here and defer the full proof to Appendix~\ref{appendix:proof_th_indexing}. We only mention that the proof follows the same strategy as that of Theorem~\ref{th:compr_setting}, with the only difference that in the indexing setting, each segment (except possibly the last one) covers at least $2\varepsilon$ keys (see~\cite[Lemma 2]{FerraginaV20}), thus reducing the possible choices for the first covered $x$- and $y$-values.

\begin{theorem}\label{th:indexing_setting}
    Let $\mathcal{I}(\ell,\varepsilon,u,n,\mathbf{x})$ be the set of PLAs of $\ell$ segments having maximum error $\varepsilon$, covering a monotone sequence of $n$ elements drawn from a universe of size $u$ in the indexing setting, and having $\mathbf{x} = x_1 < x_2 < \dots < x_{\ell}$ as the first covered $x$-values. Then:
    \begin{align*}
        \lvert \mathcal{I}(\ell,\varepsilon,u,n,\mathbf{x}) \rvert = &\binom{u-\ell(2\varepsilon-1)}{\ell}\binom{n-\ell(2\varepsilon-1)-1}{\ell-1}\left(\prod_{i=1}^{\ell-1}(x_{i+1}-x_{i}-1)\right)(2\varepsilon+1)^{2\ell}
    \end{align*}
\end{theorem}

Theorem~\ref{th:indexing_setting} gives the following lower bound on the number of bits that any lossless compressor has to produce when applied to PLAs in the indexing setting.

\begin{corollary}\label{cor:indexing_setting}
    The minimum number of bits to represent any PLA from $\mathcal{I}(\ell, \varepsilon, u, n, \mathbf{x})$ is:
    \begin{align*}
        \mathcal{B}_{\mathcal{I}}(\ell, \varepsilon, u, n, \mathbf{x}) &= \log{\binom{u-\ell(2\varepsilon-1)}{\ell}} +\log{\binom{n-\ell(2\varepsilon-1)-1}{\ell-1}} +\\ &+ \sum_{i=1}^{\ell-1}\log{(x_{i+1}-x_{i}-1)} + 2\ell\log{(2\varepsilon+1)}
    \end{align*}
\end{corollary}

As already mentioned, unlike in the compression setting, segments always cover at least $2\varepsilon$ keys (except possibly the last one). Therefore, the lower bound of Corollary~\ref{cor:indexing_setting} can be made independent from the choice of the first covered $x$-values, by noting that $x_{i+1}-x_{i} \geq 2\varepsilon$. Nevertheless, this more general form of the lower bound comes at the cost of underestimating the minimum number of bits required to store any PLA in the indexing case when segments span more than $2\varepsilon$ elements.

Having established lower bounds on the space required to represent lossless any PLA in both the compression and indexing settings, we now shift our focus to designing data structures that use space close to $\mathcal{B}_{\mathcal{C}}(\ell, \varepsilon, u, n, \mathbf{y})$ and $\mathcal{B}_{\mathcal{I}}(\ell, \varepsilon, u, n, \mathbf{x})$ while efficiently supporting the core operation of any learned data structure hinging on PLAs (e.g.,~\cite{FerraginaV20, BoffaFV22, VaidyaKCKMI22, FerraginaFMV23, GuerraVBF25, FerraginaL25}), namely, returning the $y$-value \emph{predicted} by the segment covering a given $x$-value.

\section{Succinct Data Structures for PLAs}\label{sec:succ_ds_plas}

We begin by establishing constraints on the parameters of a PLA that are commonly observed in many practical applications. Specifically, we assume that $\epsilon = O(1)$, $\ell = o(n)$ and $n = o(u)$. Next, we proceed by formalizing the core operation performed on any PLA in both the compression and indexing settings. In particular, we define the $predict(x)$ operation, which, given a value lying on the $x$-axis, identifies the segment covering $x$ and computes the corresponding $y$-value \emph{predicted} by that segment. Designing data structures that efficiently support this operation while achieving space usage close to the lower bounds of Section~\ref{sec:compr_measures_plas} is of both practical and theoretical interest since it can help to understand what the limits of PLA-based learned data structures are in terms of the redundancy needed to store and query the PLA itself, and possibly improve the performances of such data structures.

At the end of Section~\ref{sec:background}, we commented on the relation between the current literature and our lower bounds, highlighting that the storage schemes for PLAs of the LA-vector and the PGM-index are {\em compact} \cite{Navarro16}. Due to space constraints, we sketch the main result here; please refer to the appendix for the complete proofs. Starting with the LA-vector (see Section~\ref{subsec:space_close_compr}), it turns out that the storage space of Lemma~\ref{lem:la_space_time} can be expressed as $\mathcal{B}_{\mathcal{C}}(\ell,\varepsilon,u,n, \mathbf{y})$ with an additive term that is either $O(\ell\log{\frac{u}{\ell}})$ or $O(\ell\log{\frac{u}{\ell}} + \frac{n}{\log^{c}{n}})$ depending if queries are answered with a binary search or with the data structure of P\u{a}tra\c{s}cu \cite{Patrascu08}. In the first case, the overhead is always $O(\mathcal{B}_{\mathcal{C}}(\ell,\varepsilon,u,n, \mathbf{y}))$, while in the second case, that bound holds only when $\ell = \omega(n/(\log{u}\log^{c}{n}))$. As a result, the storage space of the LA-vector is asymptotically optimal up to a constant factor, and thus is compact. Moving to the PGM-index (see Appendix~\ref{appendix:bound_pgm}), the additive term on top of the lower bound $\mathcal{B}_{\mathcal{I}}(\ell,\varepsilon,u,n,\mathbf{x})$ is either $\Theta(\ell\log{u})$ or $O(\ell\log{u} + \frac{u}{\log^{c}{u}})$, depending on how queries are answered. In the first case, the overhead is always $\Theta(\mathcal{B}_{\mathcal{I}}(\ell,\varepsilon,u,n,\mathbf{x}))$, while in the second case that bound holds only when $\ell = \Omega(u/\log^{c+1}{u})$. Therefore, the storage space of the PGM-index is also compact.

Motivated by this, the following sections will describe the design of the first data structures for the compression and indexing settings, achieving small additive terms to the lower bounds of Section~\ref{sec:compr_measures_plas}, which turn out to be succinct in most practical cases, and most importantly provide the same asymptotic query time of the LA-vector and the PGM-index.






\subsection{Using space close to $\mathcal{B}_{\mathcal{C}}(\ell,\varepsilon,u,n, \mathbf{y})$}\label{subsec:space_close_compr}


Given a PLA $\{s_{1},s_{2},\dots,s_{\ell}\} \in \mathcal{C}(\ell, \varepsilon, u, n, \mathbf{y})$, we represent a generic segment $s_{i}$ as a tuple $s_{i} = (x_{i}, \beta_{i}, \gamma_{i}, y_{i}, y'_{i})$ where $x_{i}$ is the first covered $x$-value of the input sequence, $\beta_{i}$ is the intercept, $\gamma_{i}$ is the last $y$-value given by the segment, $y_{i}$ and $y'_{i}$ are respectively the first and last $y$-values of the input sequence covered by the $i$-th segment (see Figure~\ref{fig:pla_compression}). To encode each segment, we begin by storing a bit sequence $X$ containing the first $x$-value covered by each segment. Specifically, $X$ encodes in unary the deltas between the first $x$-values of consecutive segments (i.e., between $x_{i+1}$ and $x_{i}$), shifted by a constant to further reduce space usage. Such a sequence is defined as follows:

\begin{itemize}
    \item $X = 0^{x_{2}-x_{1}-2}10^{x_{3}-x_{2}-2}1\dots 0^{x_{\ell}-x_{\ell-1}-2}1$. The sequence is well-defined since $x_{i+1}-x_{i} \geq 2$. The length of $X$ is $\sum_{i=1}^{\ell-1}(x_{i+1}-x_{i}-1) = x_{\ell} - x_{1} - (\ell-1)$, which is upper bounded by $n-\ell-1$ since $x_{1} = 1$ and $x_{\ell} \leq n-1$, and contains $\ell-1$ ones. The $i$-th element of $X$ can be accessed by noticing that $x_{1} = 1$ and for any $i > 1$, $x_{i} = select_{1}(X,i-1)+i$.
\end{itemize}

Next, we store an integer sequence $Y$ encoding the first covered $y$-values of each segment. Given our assumption that $\ell = o(n)$, both the $X$ and $Y$ sequences are sparse. Since they are also monotonically increasing by construction, we represent them using the Elias–Fano encoding scheme\footnote{We turn $X$ into an integer sequence by interpreting each unary code as an integer.}. By Theorem~\ref{th:elias_fano}, the space in bits required to store $X$ and $Y$ is:

\begin{equation}\label{eq:ef_space1}
    (\ell-1) \log{\left(\frac{n-\ell-1}{\ell-1}\right)} + \ell \log{\left(\frac{u}{\ell}\right)} + 4\ell + o(\ell)
\end{equation}

Noting that $\log{\binom{u+\ell-1}{\ell}} = \log{\binom{u}{\ell}} + O(\ell)$ (since $\binom{u+\ell-1}{\ell} = \binom{u}{\ell} \cdot \left(\prod_{i=1}^{\ell-1}(u+i)/\prod_{i=1}^{\ell-1}(u-i)\right)$ and the second factor can be bounded by $c^{\ell-1}$ for some constant $c>1$) and applying the Elias-Fano space bound from Equation~\ref{eq:ef_entropy}, we can rewrite Equation~\ref{eq:ef_space1} in the following form:

\begin{equation}\label{eq:ef_space2}
    \log{\binom{n-\ell-1}{\ell-1}} + \log{\binom{u+\ell-1}{\ell}} + O(\ell+\log{u})
\end{equation}

Next we consider the sequence of last covered $y$-values of each segment, i.e. $y'_{1},y'_{2},\dots,y'_{\ell-1}$, omitting $y'_{\ell}$ as it is always equal to $u$. Since each $y'_{i} \in [y_{i},y_{i+1}]$, we store the relative values using just $\log{(y_{i+1}-y_{i}+1)}$ bits and concatenate their representation in a single bit vector $B$, of size $|B| = \sum_{i=1}^{\ell-1}\log{(y_{i+1}-y_{i}+1)}$ bits. To access each $y'_{i}$, we maintain an auxiliary array $P$, where each $p_{i}$ is the starting index of the binary representation of $y'_{i}$ inside $B$. Notice that this technique resembles the \emph{dense pointers} method of Ferragina \& Venturini~\cite{FerraginaV07}. Array $P$ is both sparse and monotonically increasing, hence we store it with the Elias-Fano encoding scheme using the following number of bits (recall Theorem~\ref{th:elias_fano}):

\begin{equation}\label{eq:space_p}
    2(\ell-1) + (\ell-1)\log{\left(\frac{\sum_{i=1}^{\ell-1}\log{(y_{i+1}-y_{i}+1)}}{\ell-1}\right)} + o(\ell) 
\end{equation}

Using Jensen's inequality, and noticing that $\sum_{i=1}^{\ell-1}(y_{i+1}-y_{i}+1) = y_{\ell}+\ell-1$ which is at most $u+\ell-1$ since $y_{\ell} \leq u$, we obtain that $|P| = O(\ell\log{\log{\frac{u}{\ell}}})$ bits. Therefore, the sequence of the last covered $y$-values is encoded as the bit vector $B$ and the Elias-Fano encoding of the array $P$, using the following number of bits:

\begin{equation}\label{eq:space_last_y}
    \left(\sum_{i=1}^{\ell-1}\log{(y_{i+1}-y_{i}+1)}\right) + O(\ell\log{\log{\frac{u}{\ell}}})
\end{equation}

Lastly, since each segment is an $\varepsilon$-approximation for the points it covers, it holds that $|\beta_{i}-y_{i}| \leq 2\varepsilon$ and $|\gamma_{i}-y'_{i}| \leq 2\varepsilon$. Therefore we store each intercept $\beta_{i}$ and each last $y$-value $\gamma_{i}$ as deltas with respect to $y_{i}$ and $y'_{i}$ inside two arrays $\Delta_{\beta}$ and $\Delta_{\gamma}$ using $2\ell\log{(2\varepsilon+1)}$ bits. Combining this with Equations~\ref{eq:ef_space2} and~\ref{eq:space_last_y}, and recalling Corollary~\ref{cor:compr_setting}, the space usage of our storage scheme is $\mathcal{B}_{\mathcal{C}}(\ell,\varepsilon,u,n, \mathbf{y}) + O(\ell\log{\log{\frac{u}{\ell}}}+\log{u})$ bits. That is succinct space, since $O(\ell\log{\log{\frac{u}{\ell}}}+\log{u}) \subseteq o(\mathcal{B}_{\mathcal{C}}(\ell,\varepsilon,u,n, \mathbf{y}))$ because $\mathcal{B}_{\mathcal{C}}(\ell,\varepsilon,u,n, \mathbf{y})$ includes a term that is $\Theta(\ell \log{\frac{u}{\ell}})$, and we consider the case that $\ell = \Omega(\log u/\log\log u)$. 

Without introducing any extra space on top of this representation, performing a $predict(x)$ query proceeds as detailed in Algorithm~\ref{alg:predict_compr}. For simplicity, we assume that the segment covering $x$ is never the first one nor the last one. Corner cases can be handled in constant time without altering the query time.

\begin{algorithm}
\SetKwInOut{Input}{Input}
\SetKwInOut{Output}{Output}
\caption{Performs a $predict$ operation on our compressed storage scheme for PLAs.}\label{alg:predict_compr}
\Input{The arrays of first covered $x$-values and $y$-values $X$ and $Y$, the bit-vector $B$, the array of offsets $P$, the array of deltas $\Delta_{\beta}$ and $\Delta_{\gamma}$, and a value $x$.}
\Output{A value $y$ such that $predict(x) = y$.}
\BlankLine
$i = pred(X,\, x)$ \tcp*{predecessor search on $X$ using a binary search}
$x_{i} = select_{1}(X,i-1)+i$\;
$x_{i+1} = select_{1}(X,i)+(i+1)$\;
$y_{i} = select_{1}(Y,\, i)$\;
$s_{i} = select_{1}(P,\,i)$\;
$e_{i} = select_{1}(P,\,i+1)-1$\;
$y'_{i} = read(B,\,s_{i},\, e_{i})$ \tcp*{read the bit sequence $B[s_{i},e_{i}]$ representing $y'_{i}$}
$\beta_{i} = y_{i} + \Delta_{\beta}[i]$\;
$\gamma_{i} = y'_{i} + \Delta_{\gamma}[i]$\;
\KwRet{$\left\lfloor (x-x_{i})\frac{(\gamma_{i}-\beta_{i})}{(x_{i+1}-x_{i})} \right\rfloor +  \beta_{i}$}\;
\end{algorithm}

Because of Theorem~\ref{th:elias_fano}, all the $select$ operations on the Elias-Fano encoded sequences $X$, $Y$, and $P$ take $O(1)$ time. Reading the binary representation of numbers from $B$, as well as accessing the deltas inside $\Delta_{\beta}$ and $\Delta_{\gamma}$, requires a $O(1)$ number of bitwise and bit-shifting operations (see~\cite[Chapter 3]{Navarro16} for the technical details), hence $O(1)$ time in the transdichotomous word RAM model we are assuming. The most expensive part of Algorithm~\ref{alg:predict_compr} is then the predecessor search on $X$, which dominates the cost of a $predict$ query, making the overall time $O(\log{\ell})$.   

To lower the cost of the predecessor search, and thus achieve faster query times, one must allow for a larger additive term to the lower bound $\mathcal{B}_{\mathcal{C}}(\ell,\varepsilon,u,n, \mathbf{y})$. To address this, we directly store the sequence of first covered $x$-values as $X = x_{2},x_{3},\dots,x_{\ell}$, whose universe is bounded by $n-1$ and contains $\ell-1$ values (recall $x_{1} = 1$ in the compression setting).

This sequence allows the predecessor of any given position (and, consequently, the index of the segment covering that position) to be computed via a simple $rank$ operation. $X$ is then stored using Theorem~\ref{th:succincter}, which guarantees that for any constant $c > 0$, both $rank$ and $select$ operations take $O(c)$ time while requiring $\log{\binom{n-1}{\ell-1}} + O\left(\frac{n}{\log^{c}{n}}\right)$ bits of space. 

Since we assumed that $\ell = o(n)$, we refactor the term $\log{\binom{n-1}{\ell-1}}$ as follows:

\begin{align*}
    \log{\binom{n-1}{\ell-1}} &= \log{\binom{n-\ell-1}{\ell-1}} + \log{\left(\frac{\prod_{i=1}^{\ell}(n-i)}{\prod_{i=\ell}^{2\ell-1}(n-i)}\right)}\\
    &\leq \log{\binom{n-\ell-1}{\ell-1}} + \ell\log{\left(\frac{n-1}{n-2\ell+1}\right)}\\
    &= \log{\binom{n-\ell-1}{\ell-1}} + O(\ell)
\end{align*}

Where the inequality is obtained by upper-bounding each term in the numerator by $n-1$ and lower-bounding each term in the denominator by $n-2\ell+1$. The space usage is then:

\begin{equation}\label{eq:compr_space_usage_pat}
    \mathcal{B}_{\mathcal{C}}(\ell,\varepsilon,u,n, \mathbf{y}) + O(\ell\log{\log{\frac{u}{\ell}}}+\frac{n}{\log^{c}n})
\end{equation}

We now show that the space usage of this second data structure is succinct. To this end, we use the well-known upper and lower bounds on the logarithm of the binomial (i.e., $k \log{\frac{m}{k}}\leq \log{\binom{m}{k}} \leq k\log{(\frac{em}{k})}$, where $e \approx 2.718$ is the Euler's number, and thus $\log{\binom{m}{k}} = \Theta(k \log{\frac{m}{k}})$), to rewrite the lower bound of Corollary~\ref{cor:compr_setting} as follows:

\begin{equation}\label{eq:compr_bound_theta}
    \mathcal{B}_{\mathcal{C}}(\ell, \varepsilon, u, n, \mathbf{y}) = \Theta(\ell(\log{\frac{u}{\ell}}+\log{\frac{n}{\ell}})) = \Theta(\ell\, \log{\frac{u}{\ell}})
\end{equation}

Where the $\Omega$-term comes from the first two additive terms in the formula of $\mathcal{B}_{\mathcal{C}}(\ell, \varepsilon, u, n, \mathbf{y})$ and the approximation of the logarithm of the binomial coefficient written above. The $O()$-term comes similarly and from upper bounding the summation in Corollary~\ref{cor:compr_setting} using Jensen's inequality, as already done in Section~\ref{subsec:space_close_compr}. 

Focusing on the additive term of Equation~\ref{eq:compr_space_usage_pat}, it is easy to see that $O(\ell\log{\log{\frac{u}{\ell}}}+\frac{n}{\log^{c}{n}}) \subseteq o(\ell\log{\frac{u}{\ell}} + \frac{n}{\log^{c}{n}})$, by considering $\ell = \omega(n/(\log{u} \, \log^{c}{n}))$ which implies that $\ell\log{\frac{u}{\ell}}  = \omega(\frac{n}{\log^{c}{n}})$. Therefore, from Equation~\ref{eq:compr_bound_theta}, it is $O(\ell\log{\log{\frac{u}{\ell}}}+\frac{n}{\log^{c}{n}}) \subseteq o(\mathcal{B}_{\mathcal{C}}(\ell, \varepsilon, u, n, \mathbf{y}))$ for those values of $\ell$ that are $\omega(n/(\log{u} \, \log^{c}{n}))$ and do not violate our initial assumption that $\ell = o(n)$.

As a result, we proved the following theorem, which gives data structures able to store any PLA in the compression setting using space close to the lower bound of Corollary~\ref{cor:compr_setting} while supporting efficient $predict$ queries.

\begin{theorem}\label{th:succ_ds_compr}
    There exist data structures storing any PLA from $\mathcal{C}(\ell, \varepsilon, u, n)$ as follows:
    \begin{itemize}
        \item Using $\mathcal{B}_{\mathcal{C}}(\ell,\varepsilon,u,n, \mathbf{y}) + O(\ell\log{\log{\frac{u}{\ell}}}+\log{u})$ bits of space, while supporting $predict$ in $O(\log{\ell})$ time.
        \item Using $\mathcal{B}_{\mathcal{C}}(\ell,\varepsilon,u,n, \mathbf{y}) + O(\ell\log{\log{\frac{u}{\ell}}} + n/\log^{c}{n})$ bits of space for any constant $c>0$, while supporting $predict$ in $O(c)$ time.
    \end{itemize}
\end{theorem}

We conclude by comparing Theorem~\ref{th:succ_ds_compr} with the storage scheme for PLAs of the LA-vector (see Lemma~\ref{lem:la_space_time}). We start by considering the one solving the $predict$ queries with a binary search on the first covered positions, reporting its space usage in bits here to facilitate comparison:

\begin{equation}\label{eq:la_vector_space_comparison}
    \ell(2\log{\frac{u}{\ell}} + \log{\frac{n}{\ell}} + 6 + 2\log{(2\varepsilon+1)} + o(1))
\end{equation}

We notice that the term $2\ell\log{(2\varepsilon+1)}$ is shared by both Corollary~\ref{cor:compr_setting} and Equation~\ref{eq:la_vector_space_comparison}, hence we only focus on the remaining parts. Since we assumed that $\ell = o(n)$, the terms $\ell \log{\frac{u}{\ell}}$ and $\ell \log{\frac{n}{\ell}}$ appearing in Equation~\ref{eq:la_vector_space_comparison} can be respectively bounded as $\log{\binom{u+\ell-1}{\ell}} + O(\ell + \log{u})$ and $\log{\binom{n-\ell-1}{\ell-1}} + O(\ell + \log{n})$ (see Section~\ref{sec:background} and~\ref{subsec:space_close_compr}). Consequently, Equation~\ref{eq:la_vector_space_comparison} can be rewritten in the following form, which is closer to the one of Corollary~\ref{cor:compr_setting}.

\begin{equation}\label{eq:la_vector_space_comparison2}
     \log{\binom{n-\ell-1}{\ell-1}} + \log{\binom{u+\ell-1}{\ell}} + \ell\log{\frac{u}{\ell}} + 2\ell\log{(2\varepsilon+1)} + O(\ell+\log{u})
\end{equation}

Lastly, the summation $\sum_{i=1}^{\ell-1}\log{(y_{i+1}-y_{i}+1)}$ appears in the lower bound of Corollary~\ref{cor:compr_setting}, but not in Equation~\ref{eq:la_vector_space_comparison2} which in turns contain also the additive terms $\ell\log{\frac{u}{\ell}} + O(\ell+\log{u})$.

Now, the minimum of the summation is $\ell-1$, and this occurs when each segment covers exactly two positions. We rule out the case in which it is zero, as this would force $y_{i+1} = y_{i}$, thus only considering PLAs formed by a single segment. In this way, the gap between the lower bound and the storage space of the LA-vector is thus $\Theta(\ell\log{\frac{u}{\ell}}) = \Theta(\mathcal{B}_{\mathcal{C}}(\ell, \varepsilon, u, n, \mathbf{y}))$. Hence, the storage scheme for PLAs of the LA-vector with the correction of Ferragina \& Lari is {\em compact}, because it is up to a constant factor from the optimal. Therefore, our first data structure in Theorem~\ref{th:succ_ds_compr} offers an even smaller additive term to the optimal space usage, while having the same asymptotic query time. Specifically, assuming $\ell = \Omega(\log{u}/\log{\log{u}})$, it exponentially reduces the overhead, from $O(\ell\log{\frac{u}{\ell}})$ to $O(\ell\log{\log{\frac{u}{\ell}}})$.

Similarly, an analogous result holds for the two data structures described in Lemma~\ref{lem:la_space_time} and Theorem~\ref{th:succ_ds_compr}, which answer queries in $O(c)$ time for any constant $c > 0$. In this case, it is sufficient to notice that term due to the data structure of P\u{a}tra\c{s}cu is shared by both, and the additive term of the LA-vector becomes $O(\ell\log{\frac{u}{\ell}}+\frac{n}{\log^{c}{n}})$. Therefore, using the same argument we used for proving the succinctness of the second data structure in Theorem~\ref{th:succ_ds_compr}, it is possible to show that for $\ell = \omega(n/\log{u}\log^{c}{n})$ the storage scheme of the LA-vector is again compact, while ours is succinct, and thus offering a lower overhead to the optimal space usage without altering its asymptotic query time.



In conclusion, we remark that, beyond its theoretical appeal, the first data structure of Theorem~\ref{th:succ_ds_compr} is also practical, as it leverages techniques with highly engineered implementations, such as the Elias-Fano encoding of monotone sequences~\cite{MaPRZ21} and arrays of fixed- and variable-size elements~\cite{Navarro16}. Indeed, under the same assumption that $\ell = \Omega(\log{u}/\log{\log{u}})$, our storage scheme asymptotically requires just $O(\log{\log{\frac{u}{\ell}}})$ bits per segment over the theoretical minimum, while still providing practically fast $predict$ queries. Conversely, the second data structure is mostly theoretical, and it aims at showing how close we can approach the lower bound of Corollary~\ref{cor:compr_setting} while still supporting (optimal) constant time $predict$ queries.

\subsection{Using space close to $\mathcal{B}_{\mathcal{I}}(\ell,\varepsilon,u,n,\mathbf{x})$}\label{subsec:space_close_indx}

The compression and indexing settings are closely related, the only difference lies in which axis is constrained to have consecutive values. Indeed, this similarity can also be observed in the two lower bounds of Corollary~\ref{cor:compr_setting} and~\ref{cor:indexing_setting} sharing the same form. As a result, the design of a data structure that achieves space usage close to the lower bound for the indexing setting, while supporting efficient $predict$ queries, builds upon the same ideas and algorithmic techniques introduced in Section~\ref{subsec:space_close_compr}. Due to space constraints, we present only the main result in this section, deferring the complete proof to Appendix~\ref{appendix:proof_indexing}.

\begin{theorem}\label{th:succ_ds_indx}
    There exist data structures storing any PLA from $\mathcal{I}(\ell, \varepsilon, u, n)$ as follows:
    \begin{itemize}
        \item Using $\mathcal{B}_{\mathcal{I}}(\ell,\varepsilon,u,n,\mathbf{x}) + O(\ell\log{\log{\frac{u}{\ell}}}+\log{u})$ bits of space, while supporting the $predict$ operation in $O(\log{\ell})$ time.
        \item Using $\mathcal{B}_{\mathcal{I}}(\ell,\varepsilon,u,n,\mathbf{x}) + O(\ell\log{\log{\frac{u}{\ell}}} + u/\log^{c}{u})$ bits of space for any constant $c>0$, while supporting the $predict$ operation in $O(c)$ time.
    \end{itemize}
\end{theorem}

The storage scheme of the PGM-index is optimal up to a constant factor. Considering the first data structure of Lemma~\ref{lem:pgm_space_time}, it can be expressed as $\mathcal{B}_{\mathcal{I}}(\ell,\varepsilon,u,n,\mathbf{x}) + \Theta(\ell\log{u})$. As a result, Theorem~\ref{th:succ_ds_indx} introduces the first data structure storing any PLA in the indexing setting within lower-order terms to $\mathcal{B}_{\mathcal{I}}(\ell,\varepsilon,u,n,\mathbf{x})$, under the assumption that $\ell = \Omega(\log{u}/\log{\log{u}})$, with the same asymptotic query time (see Appendix~\ref{appendix:bound_pgm} for a detailed comparison).

The same observations of Section~\ref{subsec:space_close_compr} regarding the practicality of our data structures hold. The first one is the most practical, while the second is just theoretical and shows how close the lower bound can be approached while answering queries in (optimal) constant time.


\section{Conclusions}\label{sec:conclusion}


In this work, we addressed the problem of deriving the first compressibility measures for PLAs of monotone sequences in two distinct settings: \emph{compression} and \emph{indexing}, which are general enough to cover a broad range of applications. The measures we obtained provide a theoretical foundation for understanding the space requirements of any learned data structure relying on PLAs in any of its components.

Based on our compressibility measures, we analyzed two representative data structures for storing PLAs, namely the one of the LA-vector (compression setting) and the one of the PGM-index (indexing setting), showing that they are both optimal up to a constant factor. 

Motivated by this study, we proposed two novel data structures for the above settings, which turn out to be succinct (thus optimal up to lower-order terms) in most practical cases while retaining the same asymptotic query time of the LA-vector and PGM-index. Overall, they offer strong theoretical guarantees and suggest efficient practical implementations, providing valuable tools for practitioners and researchers in learned data structures.

Future works could investigate an extension of our analysis to \emph{Piecewise Non-Linear Approximations}, for which successful applications already exist~\cite{GuerraVBF25}. Another interesting direction, complementary to our work, is the implementation of the proposed data structures with a thorough experimental evaluation of their space-time trade-offs.


\bibliography{lipics-v2021-sample-article}

\newpage

\appendix

\section{Detailed Proofs}\label{appendix:A}

\subsection{Complete Proof of Theorem~\ref{th:indexing_setting}}\label{appendix:proof_th_indexing}

We restate Theorem~\ref{th:indexing_setting} here to facilitate the presentation.

\begin{theorem}\label{th:indexing_setting_appendix}
    Let $\mathcal{I}(\ell,\varepsilon,u,n,\mathbf{x})$ be the set of PLAs of $\ell$ segments having maximum error $\varepsilon$, covering a monotone sequence of $n$ elements drawn from a universe of size $u$ in the indexing setting, and having $\mathbf{x} = x_1 < x_2 < \dots < x_{\ell}$ as the first covered $x$-values. Then:
    \begin{align*}
        \lvert \mathcal{I}(\ell,\varepsilon,u,n,\mathbf{x}) \rvert = &\binom{u-\ell(2\varepsilon-1)}{\ell}\binom{n-\ell(2\varepsilon-1)-1}{\ell-1}\left(\prod_{i=1}^{\ell-1}(x_{i+1}-x_{i}-1)\right)(2\varepsilon+1)^{2\ell}
    \end{align*}
\end{theorem}

\begin{proof}
    Let us denote with $x_{i}$ the first $x$-value covered by the $i$-th segment. We start by noticing that a segment always covers at least $2\varepsilon$ points (elements) in the indexing setting. This is because, in the indexing setting, points have monotonically increasing contiguous $y$-values and so the horizontal segment $y = i + \varepsilon$ is always an $\varepsilon$-approximation of any subsequence of $2\varepsilon$ consecutive elements, see~\cite[Lemma 2]{FerraginaV20}. Hence, it is $x_{i+1} - x_{i} \geq 2\varepsilon$ and $x_{\ell} \leq u - 2\varepsilon + 1$.

    We perform a similar change of variables as in the proof of Theorem \ref{th:compr_setting} by defining $\bar{x}_{i} = x_{i} - (2\varepsilon-1)(i-1)$. This guarantees that $x_{i+1} - x_{i} = \bar{x}_{i+1}-\bar{x}_{i} + 2\varepsilon-1 \geq 2\varepsilon$ because $\bar{x}_{i+1} - \bar{x}_{i} \geq 1$. Moreover, $\bar{x}_{1} \geq 1$ and $\bar{x}_{\ell} = x_{\ell} - (2\varepsilon-1)(\ell-1) \leq u - \ell(2\varepsilon-1)$, since $x_{\ell} \leq u - 2\varepsilon + 1$. Therefore, choosing the $\ell$ starting keys of each segment correspond to selecting $\ell$ distinct elements from $\{1,2,\dots,u-\ell(2\varepsilon-1)\}$. Hence, the total number of ways to perform this selection is $\binom{u-\ell(2\varepsilon-1)}{\ell}$.

    Next, since the $x$-values are not consecutive in the indexing case, we are free to choose the last covered $x$-value of each segment, which we denote by $x'_i$ (see Figure~\ref{fig:pla_indexing}). To ensure a valid PLA, we require that $x'_i \geq x_i + 1$ and $x'_i \leq x_{i+1} - 1$ for every $1 \leq i < \ell$, with $x'_{\ell} = u$. Since each $x'_i$ can be independently chosen from the interval $[x_i + 1, x_{i+1} - 1]$, the total number of possible choices is $\prod_{i=1}^{\ell-1}(x_{i+1}-x_{i}-1)$. Notice that this quantity is always larger than zero because $x_{i+1} - x_{i} \geq 2\varepsilon$ and we assumed $\varepsilon \geq 1$.

    Moving to the $y$-axis (which represents the positions of the keys), let $y_{i}$ be the position of the first key covered by the $i$-th segment and $\beta_{i}$ its intercept (see Figure~\ref{fig:pla_indexing}). Choosing $\beta_{i}$ can be reduced to first selecting $y_{i}$, and then choosing any value in the interval $[y_{i}-\varepsilon, y_{i}+\varepsilon]$ whose size is $2\varepsilon+1$. As already mentioned, in the indexing case, each segment covers at least $2\varepsilon$ keys, hence $y_{i+1}-y_{i} \geq 2\varepsilon$ and $y_{\ell} \leq n-2\varepsilon+1$. Noticing that $y_{1} = 1$ and performing the change of variables $\bar{y}_{i} = y_{i} - (2\varepsilon-1)(i-1)$, we obtain that $\bar{y}_{2} \geq 2$ and $\bar{y}_{\ell} \leq n-\ell(2\varepsilon-1)$ since $y_{2} \geq 2\varepsilon + 1$ and $y_{\ell} \leq n-2\varepsilon+1$. It follows that choosing the $\ell$ intercepts is equivalent to first choosing $\ell-1$ distinct elements in $\{2,3,\dots,n-\ell(2\varepsilon-1)\}$ and, then, for each of them, selecting a value inside an interval of size $2\varepsilon+1$. Therefore, the total number of ways to perform this choice is $\binom{n-\ell(2\varepsilon-1)-1}{\ell-1}(2\varepsilon+1)^{\ell}$.
    
    To conclude the proof, we need to consider for every $i$-th segment its slope, and thus the $y$-value on its last covered $x$-value. Let us denote this value by $\gamma_{i}$. Since we are in the indexing setting, we know that the $y$-values of the sequence are monotonically increasing and contiguous, so the last covered $y$-value is $y_{i+1} - 1$, see Figure~\ref{fig:pla_indexing}, and thus $\gamma_{i}$ can assume any value within the interval $[y_{i+1} - 1 - \varepsilon, y_{i+1} - 1 + \varepsilon]$ to make that segment a valid $\epsilon$-approximation. Therefore, we introduce a $(2\varepsilon+1)^{\ell}$ factor to the previous quantity, and thus our claim follows.  
\end{proof}

We conclude by noting that, unlike in the compression setting, in the indexing case each segment covers at least $2\varepsilon$ elements (except possibly the last one). Consequently, a lower bound to the counting of Theorem~\ref{th:indexing_setting_appendix} can be obtained by replacing $x_{i+1}-x_{i}$ with $2\varepsilon$. This yields a formula that is independent of the sequence of first covered $x$-values.

\begin{corollary}\label{cor:indexing_setting_gen_appendix}
    Let $\mathcal{I}(\ell,\varepsilon,u,n)$ be the set of PLAs of $\ell$ segments having maximum error $\varepsilon$, covering a monotone sequence of $n$ elements drawn from a universe of size $u$ in the indexing setting, then:
    \begin{align*}
    \lvert \mathcal{I}(\ell,\varepsilon,u,n) \rvert = &\binom{u-\ell(2\varepsilon-1)}{\ell}\binom{n-\ell(2\varepsilon-1)-1}{\ell-1}(2\varepsilon-1)^{\ell-1}(2\varepsilon+1)^{2\ell}
    \end{align*}
\end{corollary}

Clearly, for any choice of $\mathbf{x}$, we have $\lvert \mathcal{I}(\ell,\varepsilon,u,n) \rvert \leq \lvert \mathcal{I}(\ell,\varepsilon,u,n, \mathbf{x}) \rvert$. Therefore, a lower bound based on Corollary~\ref{cor:indexing_setting_gen_appendix} is more general, as it depends only on the sequence length $n$, the universe size $u$, the maximum error $\varepsilon$, and the number of segments $\ell$. However, this comes at the cost of underestimating the minimum number of bits required to store any PLA in the indexing setting when segments cover more than $2\varepsilon$ elements.

\subsection{PGM-index Space Usage in Terms of $\mathcal{B}_{\mathcal{I}}(\ell, \varepsilon, u, n, \mathbf{x})$}\label{appendix:bound_pgm}

As in Section~\ref{subsec:space_close_compr}, we start by considering the PGM-index storage scheme for PLAs of Lemma~\ref{lem:pgm_space_time} solving $predict$ queries with a binary search on the first covered keys, reporting its space usage in bits here to facilitate comparison:

\begin{equation}\label{eq:pgm_space_appendix}
    \ell(1.92 + \log{\frac{n}{\ell}} + \log{n} + 2\log{u} + o(1))
\end{equation}

Since we assumed $\ell = o(n)$ and $\varepsilon = O(1)$, the terms $\ell\log{\frac{n}{\ell}}$ and $\ell\log{u}$ in Equation~\ref{eq:pgm_space_appendix} can be respectively bounded as $\log{\binom{n-\ell(2\varepsilon-1)-1}{\ell-1}} + O(\ell + \log{n})$ and $\log{\binom{u-\ell(2\varepsilon-1)}{\ell}} + O(\ell\log{\ell} + \log{u})$ (see Section~\ref{sec:background} and Appendix~\ref{appendix:proof_indexing}). Noticing that, since $\varepsilon$ is a positive constant, $1.92\ell$ is equal to $2\ell\log{(2\varepsilon+1)}-\Theta(\ell)$, moreover there are some sparing terms of $\ell \log u$ (because of the factor $2$) and $\ell \log n$ in Equation~\ref{eq:pgm_space_appendix}. Hence, Equation~\ref{eq:pgm_space_appendix} can be rewritten in the following form that reminds the lower bound of Corollary~\ref{cor:indexing_setting}.

\begin{equation}\label{eq:pgm_space_appendix2}
    \log{\binom{n-\ell(2\varepsilon-1)-1}{\ell-1}} + \log{\binom{u-\ell(2\varepsilon-1)}{\ell}} + 2\ell\log{(2\varepsilon+1)} + \Theta(\ell(\log{u}+\log{n}))
\end{equation}

Again, the summation $\sum_{i=1}^{\ell-1}\log{(x_{i+1}-x_{i}-1)}$ appearing in Corollary~\ref{cor:indexing_setting} but not in Equation~\ref{eq:pgm_space_appendix2} can be handled using the same argument as in Section~\ref{subsec:space_close_compr}. Therefore, the space usage of the storage scheme for PLAs of the PGM-index can be expressed as follows:

\begin{equation*}
    \mathcal{B}_{\mathcal{I}}(\ell, \varepsilon, u, n, \mathbf{x}) + \Theta(\ell(\log{u} + \log{n}))
\end{equation*}

Notice that $\Theta(\ell(\log{u} + \log{n})) = \Theta(\ell \log{u}) \not\subseteq o(\mathcal{B}_{\mathcal{I}}(\ell, \varepsilon, u, n, \mathbf{x}))$, because of the presence in $\mathcal{B}_{\mathcal{I}}(\ell, \varepsilon, u, n, \mathbf{x})$ of terms such as $\ell\log{\frac{n}{\ell}}$ and $\ell\log{\frac{u}{\ell}}$. Therefore, the storage scheme for PLAs of the PGM-index is {\em compact}, but not succinct. Instead, our data structure of Theorem~\ref{th:succ_ds_indx} is succinct, offering a significantly smaller overhead to the optimal space usage while providing the same asymptotic query time.

Finally, the same result can be achieved when considering the data structure of Lemma~\ref{lem:pgm_space_time}, answering queries in $O(c)$ time for any given constant $c > 0$. Such a data structure is compact for $\ell = \Omega(u/\log^{c+1}{n})$, while ours is succinct and provides the same query time.

\subsection{Using Space Close to $\mathcal{B}_{\mathcal{I}}(\ell,\varepsilon,u,n,\mathbf{x})$}\label{appendix:proof_indexing}

Given a PLA $\{s_{1},s_{2},\dots,s_{\ell}\} \in \mathcal{I}(\ell, \varepsilon, u, n)$, we proceed similarly as in the compression setting. We represent a generic segment $s_{i}$ as a tuple $s_{i} = (x_{i}, x'_{i}, \beta_{i}, \gamma_{i}, y_{i})$. $x_{i}$ and $x'_{i}$ are the first and last covered $x$-values, $\beta_{i}$ is the intercept, $\gamma_{i}$ is the last $y$-value given by the segment, while $y_{i}$ is the first covered $y$-value. To encode each segment, To encode each segment, we first store two bit sequences, $X$ and $Y$, which respectively represent the first $x$-value and the last $y$-value covered by each segment. The sequence $X$ encodes, in unary, the deltas between the $x$-values of consecutive segments (i.e., $x_{i+1}$ and $x_i$), shifted by a constant to further reduce space usage. This representation enables constant time access to any value via a simple $select$ operation. An identical strategy is applied to the sequence $Y$, which encodes the deltas between successive $y$-values. Such sequences are detailed as follows.

\begin{itemize}
    \item $X = 0^{x_{1}-1}10^{x_{2}-x_{1}-2\varepsilon}1\dots0^{x_{\ell}-x_{\ell-1}-2\varepsilon}1$. The sequence is well defined since $x_{1} \geq 1$ and $x_{i+1}-x_{i} \geq 2\varepsilon$. The length of $X$ is $x_{1} + \sum_{i=1}^{\ell-1}(x_{i+1}-x_{i}-2\varepsilon+1) = x_{\ell} - (\ell-1)(2\varepsilon-1)$, which is indeed upper bounded by $u-\ell(2\varepsilon-1)$ since $x_{\ell} \leq u-2\varepsilon+1$, and contains $\ell$ ones. The $i$-th element of $X$ can be accessed by noticing that $x_{i} = select_{1}(X, i)+2\varepsilon(i-1)$.

    \item $Y = 0^{y_{2}-y_{1}-2\varepsilon}1\dots 0^{y_{\ell}-y_{\ell-1}-2\varepsilon}1$. The sequence is well defined since $y_{i+1}-y_{i} \geq 2\varepsilon$. The length of $Y$ is $\sum_{i=1}^{\ell-1}(y_{i+1}-y_{i}-2\varepsilon+1) = y_{\ell} - y_{1} - (\ell-1)(2\varepsilon-1)$, which is upper bounded by $n-\ell(2\varepsilon-1)-1$ as $y_{1} = 1$ and $y_{\ell} \leq n - 2\varepsilon+1$, and contains $\ell-1$ ones. Any element of $Y$ can be retrieved by noticing that $y_{1} = 1$ and for any $i > 1$, $y_{i} = select_{1}(Y, i-1) + 2\varepsilon(i-1)$.
\end{itemize}

Again, given our assumption that $\ell = o(n)$, both sequences are sparse. Since they are also monotonically increasing by construction, we represent them using the Elias–Fano encoding scheme. By Theorem~\ref{th:elias_fano}, using the same argument as in the compression setting, the space required to store $X$ and $Y$ is:

\begin{equation}\label{eq:ef_space_i}
    \log{\binom{u-\ell(2\varepsilon-1)}{\ell}} + \log{\binom{n-\ell(2\varepsilon-1)-1}{\ell-1}} + O(\ell+\log{u})
\end{equation}

Next, we focus on the sequence of last covered keys of each segment, namely $x'_{1},x'_{2},\dots,x'_{\ell-1}$, we do not consider $x'_{\ell}$ as it is always equal to $u$. Since each $x'_{i} \in [x_{i}+1,x_{i+1}-1]$, we exploit the same technique as in the compression setting using the following number of bits:

\begin{equation}\label{eq:space_last_y_indx}
    \left(\sum_{i=1}^{\ell-1}\log{(x_{i+1}-x_{i}-1)}\right) + O(\ell\log{\log{\frac{u}{\ell}}})
\end{equation}

Unlike the compression setting, the $y$-values are consecutive; hence, the position of the last covered value of the $i$-th segment is implicitly given by $y_{i+1}-1$. Moreover, since each segment is a $\varepsilon$-approximation for the points it covers, $|\beta_{i}-y_{i}| \leq 2\varepsilon$ and $|\gamma_{i}-(y_{i+1}-1)| \leq 2\varepsilon$. Therefore, we store each intercept $\beta_{i}$ and last $y$-value $\gamma_{i}$ as deltas with respect to $y_{i}$ and $y_{i+1}-1$ inside two arrays $\Delta_{\beta}$ and $\Delta_{\gamma}$ using $2\ell\log{(2\varepsilon+1)}$ bits. Notice that, since $y_{\ell+1}$ is not available, we store $\gamma_{\ell}$ directly using only $O(\log{n})$ bits. Combining this with Equation~\ref{eq:ef_space_i} and~\ref{eq:space_last_y_indx}, and recalling Corollary~\ref{cor:indexing_setting}, the space usage of our storage scheme can be expressed as follows:

\begin{equation}
    \mathcal{B}_{\mathcal{I}}(\ell,\varepsilon,u,n,\mathbf{x}) + O(\ell\log{\log{\frac{u}{\ell}}} + \log{u})
\end{equation}

That is succinct space, since $O(\ell\log{\log{\frac{u}{\ell}}} + \log{u}) \subseteq o(\mathcal{B}_{\mathcal{I}}(\ell,\varepsilon,u,n,\mathbf{x}))$ because $\mathcal{B}_{\mathcal{I}}(\ell,\varepsilon,u,n,\mathbf{x})$ includes a term that is $\Theta(\ell \log{\frac{u}{\ell}})$, and we consider the case that $\ell = \Omega(\log u/\log\log u)$. Without adding any extra space on top of this representation, performing a $predict(x)$ query proceeds exactly as in the compression setting. Hence, after computing the index $i$ of the segment covering $x$, we retrieve its the parameters $x_{i}$, $\beta_{i}$, $\gamma_{i}$, and $x'_{i}$ using a procedure that closely resembles the one already detailed in Section~\ref{subsec:space_close_compr}, thus returning:

\begin{equation*}
    predict(x) = \left\lfloor (x-x_{i})\frac{(\gamma_{i}-\beta_{i})}{(x'_{i}-x_{i})} \right\rfloor +  \beta_{i}
\end{equation*}

The cost of a $predict$ is again dominated by the predecessor search on $X$, which overall makes the query time $O(\log{\ell})$. Still, it is possible to reduce this time by directly storing the sequence $X = x_{1},x_{2},\dots,x_{\ell}$, whose universe is bounded by $u-(2\varepsilon-1)$ and contains $\ell$ ones. The index of the segment covering a given value $x$ can be computed with a $rank$ operation, and it allows access to any $x_{i}$ using a $select$ operation. Applying Theorem~\ref{th:succincter} on $X$ guarantees that for any constant $c > 0$, both $rank$ and $select$ operations take $O(c)$ time while requiring the following space usage in bits: 

\begin{equation*}
    \log{\binom{u-(2\varepsilon-1)}{\ell}} + O\left(\frac{u}{\log^{c}{u}}\right)
\end{equation*}

Since we assumed that $\varepsilon = O(1)$ and $\ell = o(n)$, hence $\ell = o(u)$, we refactor the term $\log{\binom{u-(2\varepsilon-1)}{\ell}}$ as follows:

\begin{align*}
    \log{\binom{u-(2\varepsilon-1)}{\ell}} &= \log{\binom{u-\ell(2\varepsilon-1)}{\ell}} + \log{\left(\frac{\prod_{i=(2\varepsilon-1)}^{\ell(2\varepsilon-1)-1}(u-i)}{\prod_{i=(2\varepsilon-1)}^{\ell(2\varepsilon-1)-1}(u-i-\ell)}\right)}\\
    &\leq \log{\binom{u-\ell(2\varepsilon-1)}{\ell}} + (\ell-1)(2\varepsilon-1)\log{\left(\frac{u-(2\varepsilon-1)}{u-(\ell+1)(2\varepsilon-1)+1}\right)}\\
    &= \log{\binom{u-\ell(2\varepsilon-1)}{\ell}} + O(\ell)
\end{align*}

Where the inequality is obtained by upper-bounding each term in the numerator by $u-(2\varepsilon-1)$ and lower-bounding each term in the denominator by $u-(\ell+1)(2\varepsilon-1)+1$. Thus, the overall space usage can be expressed as:

\begin{equation}
    \mathcal{B}_{\mathcal{I}}(\ell,\varepsilon,u,n,\mathbf{x}) + O(\ell\log{\log{\frac{u}{\ell}}} + \frac{u}{\log^{c}u})
\end{equation}

Again, this is succinct space since $O(\ell\log{\log{\frac{u}{\ell}}} + \frac{u}{\log^{c}u}) \subseteq o(\mathcal{B}_{\mathcal{I}}(\ell,\varepsilon,u,n,\mathbf{x}))$ as long as $\ell = \Omega(u/\log^{c}{u})$ (the proof is similar to the one in Section~\ref{subsec:space_close_compr}). As a result, we proved the equivalent of Theorem~\ref{th:succ_ds_compr} in the indexing setting, which gives the first data structures able to store any PLA in such a scenario using space close to the information-theoretic lower bound of Corollary~\ref{cor:indexing_setting} while still supporting efficient $predict$ operations.

\begin{theorem}\label{th:succ_ds_indx_appendix}
    There exist data structures storing any PLA from $\mathcal{I}(\ell, \varepsilon, u, n)$ as follows:
    \begin{itemize}
        \item Using $\mathcal{B}_{\mathcal{I}}(\ell,\varepsilon,u,n,\mathbf{x}) + O(\ell\log{\log{\frac{u}{\ell}}}+\log{u})$ bits of space, while supporting the $predict$ operation in $O(\log{\ell})$ time.
        \item Using $\mathcal{B}_{\mathcal{I}}(\ell,\varepsilon,u,n,\mathbf{x}) + O(\ell\log{\log{\frac{u}{\ell}}} + u/\log^{c}{u})$ bits of space for any constant $c>0$, while supporting the $predict$ operation in $O(c)$ time.
    \end{itemize}
\end{theorem}

Finally, we redirect to Section~\ref{subsec:space_close_indx} for a detailed discussion on how our data structures relate to the storage scheme for PLAs of the PGM-index of Lemma~\ref{lem:pgm_space_time}.

\end{document}